\documentclass[AMS,STIX1COL]{WileyNJD-v2}
\usepackage{amsmath,amssymb,amsfonts,amsthm}
\usepackage{algorithm}
\usepackage[noend]{algcompatible}
\usepackage{graphicx}
\usepackage{textcomp}
\usepackage{xcolor}
\usepackage{subfigure}
\usepackage{todonotes}
\articletype{Research Article}%

\received{}
\revised{}
\accepted{}

\newtheorem{defn}{Definition}
\newtheorem{assump}{Assumption}
\newtheorem{prop}{Proposition}
\newtheorem{thm}{Theorem}
\newtheorem{prob}{Problem}
\newtheorem{rem}{Remark}
\newtheorem{cor}{Corollary}

\begin{document}

\title{Cluster Assignment in Multi-Agent Systems : Sparsity Bounds and Fault Tolerance\protect\thanks{This research was supported by grant no. 2285/20 from the Israel Science Foundation.}}

\author[1]{Miel Sharf*}

\author[2]{Daniel Zelazo}

\authormark{Sharf and Zelazo}

\address[1]{\orgname{Jether Energy Research}, \orgaddress{\state{Tel Aviv}, \country{Israel}}}

\address[2]{\orgdiv{Faculty of Aerospace Engineering}, \orgname{Technion - Israel Institute of Technology}, \orgaddress{\state{Haifa}, \country{Israel}}}

\corres{*Miel Sharf, \email{mielsharf@gmail.com}}


\abstract[Summary]{We study cluster assignment in homogeneous diffusive multi-agent networks. Given the number of clusters and agents within each cluster, we design the network graph ensuring the system will converge to the prescribed cluster configuration. Using recent results linking clustering and symmetries,  we show that it is possible to design an oriented graph for which the action of the automorphism group of the graph has orbits of predetermined sizes, guaranteeing the network will converge to the prescribed cluster configuration. We provide bounds on the number of edges needed to construct these graphs along with a constructive approach for their generation. We also consider the robustness of the clustering process under agent malfunction.} 


\maketitle

\footnotetext{\textbf{Abbreviations:} MAS, Multi-Agent System(s);}

\section{Introduction}\label{sec.intro}
One of the most important tasks in the field of multi-agent systems (MAS) is reaching agreement. Distributed protocols guaranteeing the agents reach agreement appear in many different fields, including robotics \cite{Chopra2006}, sensor networks \cite{OlfatiSaber2007cdc}, and distributed computation \cite{Xiao2004}. A natural generalization is the \emph{cluster agreement} problem, which seeks to drive agents into groups, so that agents within the same group reach an agreement. The clustering problem appears in social networks \cite{Lancichinetti2012}, neuroscience \cite{Schnitzler2005}, and biomimetics \cite{Passino2002}.
Clustering has been studied using different approaches, e.g. network optimization \cite{Burger2011TAC}, pinning control \cite{Qin2013}, inter-cluster nonidentical inputs \cite{Han2013}, and exploiting the structural balance of the underlying graph \cite{Altafini2013}. We tackle the cluster agreement problem by using \emph{symmetries} within the MAS. Recent works on MAS apply graph symmetries to study various problems, e.g., controllability and observability of MAS \cite{Rahmani2007, Chapman2014, Chapman2015, Notarstefano2013TAC}. 

Our recent work \cite{Sharf2019b} introduced the notion of the \emph{weak automorphism group of a MAS}, combining the two concepts of the automorphism group for graphs and weak equivalence of dynamical systems.
The former summarizes all symmetries for a given graph, while the latter characterizes similarities between achievable steady-states of heterogeneous (dynamical) agents. Therefore, the weak automorphisms of a MAS can be  understood as permutations of the nodes in the underlying graph that preserve both graph symmetries and certain input-output properties of the corresponding agents. More specifically, \cite{Sharf2019b} focused on clustering for \emph{diffusively coupled networks},  and showed that under appropriate passivity assumptions, these diffusively coupled networks converge to a clustered steady-state solution, where two agents are in the same cluster if and only if there exists a weak automorphism mapping one to the other. Thus, the clustering of the MAS can be understood by studying the action of the weak automorphism group on the underlying interaction graph. 

We focus in this paper on homogeneous networks, i.e., networks where the agent dynamics are all identical, noting that the weak automorphism group is identical to the automorphism group of the underlying graph in this case. We wish to design graphs ensuring the MAS will converge to a prescribed cluster configuration, i.e., specifying the number of clusters and the number of agents within each cluster. \textcolor{black}{Our previous work, \cite{Sharf2022a},} applied tools from group theory to prescribe an algorithm for constructing an oriented graph such that the action of the automorphism group on the graph has orbits of prescribed sizes. \textcolor{black}{It also} provided upper and lower bounds on the number of edges needed to construct such graph. \textcolor{black}{This work extends the previous work \cite{Sharf2022a} in two ways. First, we further explore the bounds on the number of edges, by understanding the reason for the difference between the upper and the lower bounds (the example in Remark \ref{rem.TreeOrPath}) as well as the scaling properties of the upper bound (Theorem \ref{thm.GlobalUpperBound} and Remark \ref{rem.WorstCaseSparse}). Furthermore, we study the robustness of such graphs to agent malfunctions, and alter the synthesis procedure to guarantee the most extensive possible robustness of the clustering possible (Subsection \ref{subsec.Robust}).}

The rest of paper is organized as  follows. Section \ref{sec.background} reviews basic concepts related to network systems and group theory required to define a notion of symmetry for MAS. Section \ref{sec.cluster} presents the main results about cluster assignment, as well as a numerical study to demonstrate the theory. Finally, some concluding remarks are offered in Section \ref{sec.conclusion}.

\paragraph*{Notations}
This work employs basic notions from graph theory \cite{Godsil2001}.  An undirected graph $\mathcal{G}=(\mathbb{V},\mathbb{E})$ consists of finite sets of vertices $\mathbb{V}$ and edges $\mathbb{E} \subset \mathbb{V} \times \mathbb{V}$.  We denote the edge with ends $i,j\in \mathbb{V}$ as $e=\{i,j\}$. For each edge $e$, we pick an arbitrary orientation and denote $e=(i,j)$ when $i\in \mathbb{V}$ is the \emph{head} of edge and $j \in \mathbb{V}$ is its \emph{tail}. A path is a sequence of distinct nodes $v_1, v_2,\ldots , v_n$ such that $\{v_i, v_{i+1}\} \in \mathbb{E}$ for all $i$. A cycle is path $v_1,\ldots,v_n,v_1$. A simple cycle is a cycle whose vertices are all distinct. A graph is connected if there is a path between any two vertices, and a tree if it is connected but contains no simple cycles.
The incidence matrix of $\mathcal{G}$, denoted $\mathcal{E}\in\mathbb{R}^{|\mathbb{E}|\times|\mathbb{V}|}$, is defined such that for any edge $e=(i,j)\in \mathbb{E}$, $[\mathcal{E}]_{ie} =+1$, $[\mathcal{E}]_{je} =-1$, and $[\mathcal{E}]_{\ell e} =0$ for $\ell \neq i,j$. 
Moreover, the greatest common divisor of two positive integers $m,n$ is denoted by $\gcd(m,n)$, and their least common multiple is denoted by $\mathrm{lcm}(m,n)$. Note that $\mathrm{lcm}(m,n) = \frac{mn}{\gcd(m,n)}$ always holds. Two integers are relatively prime (or coprime) if there is no integer greater than one that divides them both. The cardinality of a finite set $A$ is denoted by $|A|$.

\vspace{-15pt}
\section{Symmetries in Networked Systems}\label{sec.background}
In this section, we provide background on the notion of symmetries for multi-agent systems originally proposed in~\cite{Sharf2019b}.

\vspace{-15pt}
\subsection{Diffusively Coupled Networks}
This section describes the structure of the MAS studied in \cite{Burger2014, Sharf2019b}. Consider a set of agents interacting over a network $\mathcal{G}=(\mathbb{V},\mathbb{E})$. Each node $i\in \mathbb{V}$ is assigned a dynamical system $\Sigma_i$, and the edges $e\in\mathbb{E}$ are assigned controllers $\Pi_e$, having the following form:
\begin{align} \label{Dynamics}
\Sigma_i: 
\begin{cases}
\dot{x}_i = f_i(x_i,u_i), ~
y_i = h_i(x_i,u_i),
\end{cases} \quad
\Pi_e: 
\begin{cases}
\dot{\eta}_e= \phi_e(\eta_e,\zeta_e), ~
\mu_e = \chi_e(\eta_e,\zeta_e)
\end{cases}.
\end{align}

\begin{figure}
        \centering
        \includegraphics[scale=0.7]{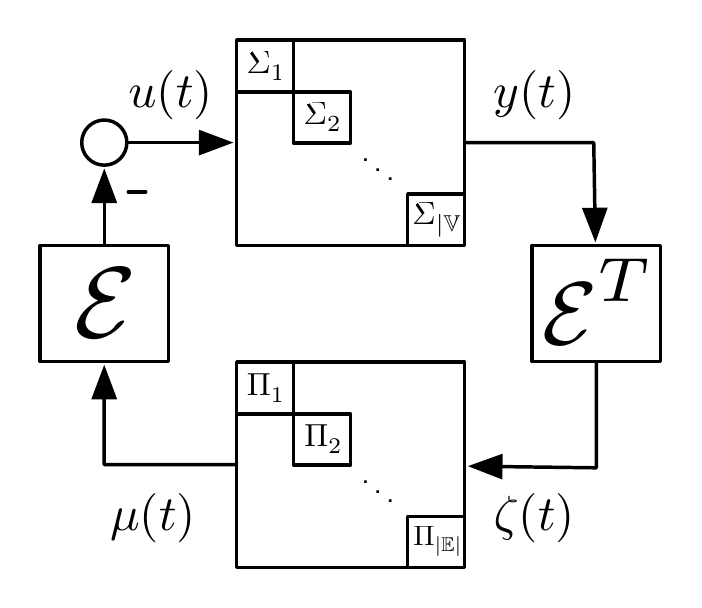}
        \vspace{-10pt}
        \caption{A diffusively coupled network.}
        \vspace{-10pt}
        \label{ClosedLoop}
\end{figure}

We consider the stacked vectors $u=[u_1^T,...,u_{|\mathbb{V}|}^T]^T$ and similarly for $y,\zeta$ and $\mu$.  
The MAS is diffusively coupled with the controller input described by $\zeta = \mathcal{E}^Ty$, and the control input is $u = -\mathcal{E}\mu$, where $\mathcal{E}$ is the incidence matrix of the graph $\mathcal{G}$.  This structure is denoted by the triplet $(\mathcal{G},\Sigma,\Pi)$, and is illustrated in Fig. \ref{ClosedLoop}.
In \cite{Burger2014} it was shown that 
the network converges to a steady-state satisfying the interconnection constraints if the agents and controllers are (output-strictly) maximum equilibrium independent passive (MEIP).  
The details of this definition and related definitions are not essential for the development of this work, and the interested reader is referred to \cite{Burger2014, Sharf2017} for more details. 
In this work, we focus on homogeneous networks, i.e., where all the agent dynamics and control dynamics are identical.
Moreover, we assume one of the following two alternatives (Assumption \ref{assump.1}).  If this is not the case, see \cite{Jain2018, Sharf2019a, Sharf2019c} for plant augmentation techniques. 

\begin{assump}\label{assump.1}
The agents $\Sigma_i$ are output-striclty MEIP and the controllers $\Pi_e$ are MEIP, or vice versa.
\end{assump}


A final technical definition is needed to characterize the steady-states of the network. Indeed, we implicitly assume that each agent and controller converges to a steady-state output given a constant input. This allows us to define a relation between constant inputs and constant outputs called the \emph{steady-state relation} of a system; see \cite{Burger2014}. We denote the steady-state relations of node $i$ and edge $e$ by $k_i$ and $\gamma_e$, respectively. For example, for an agent $i$, we say that $(\mathrm{u}_i,\mathrm{y}_i)$ is a steady-state input/output pair if $\mathrm{y}_i \in k_i(\mathrm{u})$.  We now introduce the notion of weak equivalence for dynamical systems.

\begin{defn}[\hspace{0.01pt}\cite{Sharf2019b}]
Two systems $\Sigma_1$ and $\Sigma_j$ are \emph{weakly equivalent} if their steady-state relations are identical.
\end{defn}
We refer the reader \cite{Sharf2019b} for a more thorough study and examples of weakly equivalent systems.

\vspace{-15pt}
\subsection{Group Theory, Graph Automorphisms, and Symmetric MAS}
Our approach for clustering will hinge on symmetry, which is modelled by the mathematical theory of groups \cite{Dummit2004}. The notion of a group can be defined in various ways, but we opt for the most concrete one.
\begin{defn}
Let $X$ be a set, and let $\mathbb{G}$ be a collection of invertible functions $X\to X$. The collection $\mathbb{G}$ is called a \emph{group} if for any $\mathbb{G} \ni f,g: X\to X$, both the composite function $f\circ g$ and the inverse function $f^{-1}$ belong to $\mathbb{G}$.
\end{defn}
Colloquially, the group $\mathbb{G}$ defines a collection of symmetries of the set $X$, and its action on $X$ allows us to identify certain elements of $X$ which are symmetric. 
Of interest in this work is the automorphism group of a (oriented) graph, which encodes structural symmetries of a graph. 
\begin{defn}
An automorphism of a \textcolor{black}{(directed or undirected)} graph $\mathcal{G} = (\mathbb{V},\mathbb{E})$ is a permutation $\psi:\mathbb{V}\to\mathbb{V}$ such that $i\in\mathbb{V}$ is connected to $j\in\mathbb{V}$ if and only if $\psi(i)$ is connected to $\psi(j)$. We denote the automorphism group of $\mathcal{G}$ by $\mathrm{Aut}(\mathcal{G})$. 
\end{defn}
We slightly abuse  notation and say that $\mathrm{Aut}(\mathcal{G})$ acts on $\mathcal{G}$ (rather than on $\mathbb{V}$).

\begin{defn}
Let $\mathbb{G}$ be a group of functions $X\to X$. We say $x,y\in X$ are \emph{exchangeable} (under the action of $\mathbb{G}$) if there is some $f\in \mathbb{G}$ such that $f(x)=y$. The \emph{orbit} of $x\in X$ is the set of elements which are exchangeable with $X$.
\end{defn}
Exchangeability was considered in \cite{Sharf2019b} to describe the clustering behavior of a MAS. Namely, the different clusters corresponded to the different orbits of the weak automorphism group of the MAS.   The following result ensures that we can consider orbits, and consequently different clusters in an MAS, as disjoint sets. 
\begin{prop}[\cite{Dummit2004}]
Let $\mathbb{G}$ be a group of functions $X\to X$. Then $X$ can be written as the union of disjoint orbits. \textcolor{black}{In particular, any element of $X$ belongs to exactly one orbit.}
\end{prop}

%
Finally, we combine the notions of graph automorphisms and diffusively-coupled MAS comprised of weakly equivalent agents.  

\begin{defn}[\cite{Sharf2019b, Sharf2022a}]\label{weakMASaut}
Let $(\mathcal{G},\Sigma,\Pi)$ be a diffusively coupled MAS. A \emph{weak automorphism of a MAS} is a map $\psi:\mathbb{V}\to\mathbb{V}$ with the following properties:
$1$) $\psi$ is an automorphism of the graph $\mathcal{G}$, and preserves edge orientations; 
$2$) for any $i\in\mathbb{V}$, $\Sigma_i$ and $\Sigma_{\psi(i)}$ are weakly equivalent; and
$3$) for any $e \in\mathbb{E}$, $\Pi_e$ and $\Pi_{\psi(e)}$ are weakly equivalent. 
We denote the collection of all weak automorphisms of $(\mathcal{G},\Sigma,\Pi)$ by $\mathrm{Aut}(\mathcal{G},\Sigma,\Pi)$. 
\end{defn}
Naturally, the weak automorphism of a MAS is a subgroup of the group of automorphisms $\mathrm{Aut}(\mathcal{G})$ of the graph $\mathcal{G}$.
\vspace{-20pt}
\section{Cluster Assignment in MAS}\label{sec.cluster}

We now consider the clustering problem for MAS. Specifically, we focus on the case where the agents are homogeneous, i.e., they have the exact same model, and restrict ourselves by requiring the edge controllers \eqref{Dynamics} are also homogeneous.
The paper \cite{Sharf2019b} established a link between the clustering behaviour of a MAS and certain symmetries it has, using Definition \ref{weakMASaut}. The main result from \cite{Sharf2019b} is summarized below.

\begin{thm}[\hspace{0.01pt}\cite{Sharf2019b}] \label{thm.Symmetry}
Take a diffusively-coupled MAS $(\mathcal{G},\Sigma,\Pi)$ where Assumption \ref{assump.1} holds. Then for any steady-state $\mathrm y=\begin{bmatrix} \mathrm y_1 &  \cdots & \mathrm y_{|\mathbb{V}|}\end{bmatrix}^T$ of the closed-loop and any weak automorphism $\psi\in\mathrm{Aut}(\mathcal{G},\Sigma,\Pi)$, we have $P_\psi \mathrm y = \mathrm y$, where $P_\psi$ is the permutation matrix representation  of $\psi$.
\end{thm}

This result can in fact be used to show that the network converges to a clustering configuration, where the clusters are given by the orbits of the weak automorphism group. Namely, one considers diffusively-coupled MAS $(\mathcal{G},\Sigma,\Pi)$ satisfying Assumption \ref{assump.1}, for which the closed-loop is known to converge, and the invariance properties of the steady-state limit are studied. Focusing on homogeneous networks, \cite{Sharf2019b} identified the value of $\gamma_e(0)$, the steady-state relation for the controller on the $e$th edge,  as indicative of clustering. Namely, it shows that if $0\in \gamma_e(0)$ for all $e\in \mathbb{E}$, then the MAS $(\mathcal{G},\Sigma,\Pi)$ converges to consensus, and otherwise it displays a clustering behavior. Namely, for homogeneous MAS, two nodes are in the same cluster if they are exchangable under the action of $\mathrm{Aut}(\mathcal{G})$, and the converse is almost surely true.

Although \cite{Sharf2019b} presented a strong link between symmetry and clustering in MAS, it did not consider a synthesis procedure for solving the clustering problem:
\begin{prob} \label{prob.cluster}
Consider a collection of $n$ homogeneous agents $\{\Sigma_i\}_{i\in \mathbb{V}}$, and let $r_1,\ldots,r_k$ be positive integers which sum to $n$. Find a directed graph $\mathcal{G} = (\mathbb{V},\mathbb{E})$ and homogeneous edge controllers $\{\Pi_e\}_{e\in \mathbb{E}}$ such that the closed loop MAS converges to a clustered steady-state, with a total of $k$ clusters of sizes $r_1,\ldots,r_k$.
\end{prob}

The goal of this section is use the tools of \cite{Sharf2019b} to solve Problem \ref{prob.cluster}. As described above, this can be achieved in two steps. We first make the following assumption about the controllers:
\begin{assump}\label{assump.3}
The homogeneous MEIP controllers are chosen so that $0\not\in\gamma_e(0)$ for any edge $e\in \mathbb{E}$.
\end{assump}
Second, given the desired cluster sizes $r_1,\ldots,r_k$, we seek an oriented graph $\mathcal{G} = (\mathbb{V},\mathbb{E})$ such that the action of $\mathrm{Aut}(\mathcal{G})$ on $\mathcal{G}$ has orbits of sizes $r_1,\ldots,r_k$. Moreover, we require $\mathcal{G}$ to be weakly connected\footnote{Recall that a directed graph is weakly connected if its unoriented counterpart is connected.} to assure a flow of information throughout the corresponding diffusively-coupled network. If we find such a graph and Assumption \ref{assump.3} holds, the results of \cite{Sharf2019b} guarantee that the desired clustering behavior is achieved almost surely. We make the following definition for the sake of brevity, and define the corresponding problem:
\begin{defn}
The oriented graph $\mathcal{G}$ is said to be of type OS$(r_1,\ldots,r_k)$ if it is weakly connected and the action of $\mathrm{Aut}(\mathcal{G})$ on $\mathcal{G}$ has orbits of sizes $r_1,\ldots,r_k$.\footnote{OS stands for "orbit structure."}
\end{defn}
\begin{prob}
Given positive integers $r_1,\ldots,r_k$, determine if an oriented graph of type OS$(r_1,\ldots,r_k)$ exists, and if so, construct it.
\end{prob}

\textcolor{black}{Before moving on to the solving this problem, we present a tool we apply later in the proofs, called the graph quotient.}
\textcolor{black}{
\begin{defn}
Let $\mathcal{G} = (\mathbb{V},\mathbb{E})$ be a (directed or undirected) graph, and let $V_1,V_2,\ldots, V_k$ be a partition of $\mathbb{V}$ to disjoint sets. The \emph{quotient of $\mathcal{G}$}, according to the partition $V_1,V_2,\ldots, V_k$, is a graph $\mathcal{C}$ with the following properties:
\begin{itemize}
    \item[i)] The nodes of $\mathcal{C}$ are denoted by $1,2,\ldots,k$.
    \item[ii)] For any $l_1,l_2 \in \{1,2,\ldots,k\}$, There is an edge $l_1\to l_2$ in the quotient graph $\mathcal{C}$ if and only if there is at least one edge between elements of $V_{l_1}$ and $V_{l_2}$.
\end{itemize}
\end{defn}
In other words, the quotient graph is achieved by grouping the nodes of $\mathcal{G}$ by the sets $V_1,V_2,\ldots, V_k$, removing edges within the same set, and identifying all edges going between the same two groups $V_i$ and $V_j$. An illustration can be seen in Fig. \ref{fig.Condensation}. It is easy to see that if $\mathcal{G}$ is connected, then so is its quotient, and this fact will play a vital role later.
}

\begin{figure}[t]
    \centering
    \includegraphics[width = 0.7\textwidth]{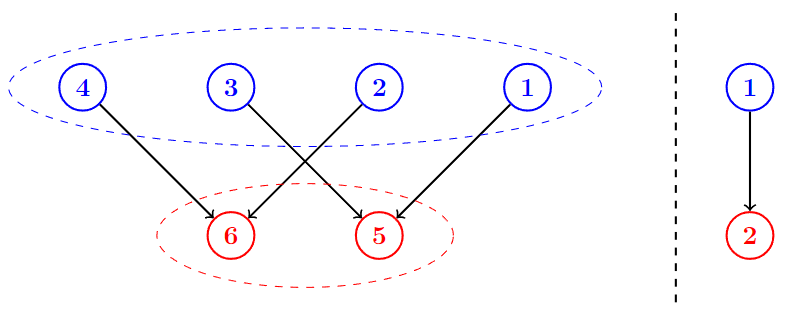}
    \vspace{-15pt}
    \caption{An example of graph quotient. The original graph $\mathcal{G}$ is depicted on left, and the sets $V_1,V_2$ of the partition are marked in blue and red, respectively. The corresponding quotient graph can be seen on the right.}
    \vspace{-15pt}
    \label{fig.Condensation}
\end{figure}

\vspace{-15pt}
\subsection{Construction and Sparsity Bounds on OS-type Graphs}
In this subsection, we exhibit a construction for OS-type graphs, as well as bounds on the sparsity of such graphs. Running the system requires means to implement the corresponding interconnections, ergo graphs with fewer edges are desirable. The following theorem provides a lower bound on the number of edges required to construct OS-type graphs.

\begin{thm}\label{thm.LowerBound}
Let $r_1,\ldots r_k$ be any positive integers. Any directed graph of type OS$(r_1,\ldots,r_k)$ has at least $m$ edges, where 
\begin{align} \label{eq.LowerBound}
m = \min_{\mathcal{T}\text{\emph{ tree on $k$ vertices}}} \sum_{\{i,j\}\in\mathcal{T}} \mathrm{lcm}(r_i,r_j).
\end{align}
\end{thm}

\begin{proof}
Let  $\mathcal{G}$ be a graph of type OS($r_1,\ldots,r_k$), and $V_1,\ldots,V_k$ be the orbits of $\mathrm{Aut}(\mathcal{G})$ in $\mathcal{G}$, corresponding to the different clusters. \textcolor{black}{The proof will consist of two steps. First, we show that if there exists at least one edge between $V_i$ and $V_j$, then there are at least $\mathrm{lcm}(r_i,r_j)$ edges between $V_i$ and $V_j$. This will follow from the fact that the automorphism group of $\mathcal{G}$ can map any two nodes in each $V_i$ to one another. Second, we consider the quotient graph by the partition $V_1,\ldots,V_k$, which must be connected as $\mathcal{G}$ is connected. Hence, it must have a spanning tree, which can be used to determine which pairs $V_i,V_j$ have edges between them. These two facts together will allow us to establish the lower bound \eqref{eq.LowerBound}.}

We start with the former claim. Suppose that there is an edge between elements of $V_i,V_j$ for some $i,j\in\{1,\ldots,k\}$. Let $\mathcal{G}_{ij}$ be the induced bi-partite subgraph on the vertices $V_i\cup V_j$, i.e., only edgess between $V_i$ and $V_j$ appear in $\mathcal{G}_{ij}$. We recall that $V_j$ is invariant under $\mathrm{Aut}(\mathcal{G})$, meaning that for any node $v \in V_i$ and any automorphism $\psi \in \mathrm{Aut}(\mathcal{G})$, $x$ and $\psi(x)$ are linked to the same number of nodes in $V_j$. As all the nodes in $V_i$ can be mapped to one another using graph automorphisms, we conclude that they all have the same $\mathcal{G}_{ij}$-degree, denoted $d_i$. By reversing the roles of $V_i$ and $V_j$, we conclude that all nodes in $V_j$ have the same $\mathcal{G}_{ij}$ degree as well, which will be denoted by $d_j$. As each edge in $\mathcal{G}_{ij}$ touches one node from $V_i$ and one node from $V_j$, we conclude that the total number of edges in $\mathcal{G}_{ij}$ is equal to $d_ir_i = d_jr_j$, as the sets $V_i,V_j$ have $r_i,r_j$ nodes, respectively.  We note that by assumption, there is at least one edge between $V_i$ and $V_j$, hence $d_ir_i = d_jr_j \ge 1$. In particular, all the numbers $d_i,r_i,d_j$ and $r_j$ are positive integers.

The rest of the proof of this part of the claim follows from basic number theory and some algebra - the equation $d_ir_i = d_jr_j$, together with the fact that $d_i,d_j,r_i$ and $r_j$ are all positive integers, imply that $r_j$ divides the product $r_i d_i$, and hence the integer $\frac{r_j}{\gcd(r_i,r_j)}$ must divide $\frac{r_i}{\gcd(r_i,r_j)} d_i$. However, by definition of the greatest common divisor, $\frac{r_j}{\gcd(r_i,r_j)}$ and $\frac{r_i}{\gcd(r_i,r_j)}$ are disjoint ,meaning that $\frac{r_j}{\gcd(r_i,r_j)}$ must divide $d_i$. In particular, as $d_i$ is a positive integer, we conclude that $d_i \ge \frac{r_j}{\gcd(r_i,r_j)}$. Thus, the number of edges in $\mathcal{G}_{ij}$ is at least $r_i d_i \ge \frac{r_i r_j}{\gcd(r_i,r_j)} = \mathrm{lcm}(r_i,r_j)$, as claimed.

Now, we move to the second part of the proof. We consider the quotient $\mathcal{C}$ of $\mathcal{G}$ by the partition $V_1,\ldots,V_k$. As $\mathcal{G}$ is (weakly) connected, so is the quotient $\mathcal{C}$. In particular, there exists a spanning tree $\mathcal{T}$ for $\mathcal{C}$. By the definition of the quotient, for any two nodes $i,j$ connected in $\mathcal{T}$ (and hence in $\mathcal{C}$), there is at least one edge between elements of $V_{i}$ and $V_{j}$ in $\mathcal{G}$. However, by the first part of the proof, we conclude that there are at least $\mathrm{lcm}(r_i,r_j)$ edges between them. By summing over all connected pairs of nodes in $\mathcal{T}$, we conclude that the graph $\mathcal{G}$ has at least $ \sum_{\{i,j\}\in\mathcal{T}} \mathrm{lcm}(r_i,r_j) \ge m$ edges. This concludes the proof.
\end{proof}

Beside giving a lower bound on the number of edges in an OS-type graph, Theorem \ref{thm.LowerBound} also highlights the role of the quotient base graph $\mathcal{T}$ in the construction of such graphs. Namely $\mathcal{T}$, which was taken as a tree, determines which clusters are connected in $\mathcal{G}$. The following algorithm, Algorithm \ref{alg.BuildingGraphs}, uses this idea to construct OS-type graphs when $\mathcal{T}$ is taken as a path graph.

\begin{algorithm} [!h]
\caption{Building sparse OS-type graphs}
\label{alg.BuildingGraphs}
{\bf Input:} A collection $r_1,\ldots,r_k$ of positive integers summing to $n$, and a path $\mathcal{T}$ on $k$ vertices.\\
{\bf Output:} A graph $\mathcal{G}$ of type OS$(r_1,\ldots,r_k)$.
\begin{algorithmic}[1]
\State Let $\mathcal{G} = (\mathbb{V},\mathbb{E})$ be an empty graph.
\State For any $j=1,\ldots,k$ and $p=1,\ldots,r_j$, add a node with label $v^j_p$ to $\mathbb{V}$.
\State For any edge $\{i,j\}$ in $\mathcal{T}$ and any $p=1,\ldots,\mathrm{lcm}(r_i,r_j)$, add the edge $v^i_{p ~\mathrm{mod}~ r_i} \to v^j_{p ~\mathrm{mod}~ r_j}$ to $\mathbb{E}$
\State Compute $i^\star = \arg\min\{r_i\}$. If $r_{i^\star} = 1$, go to step 6.
\State For any $p=1,\ldots,r_{i^\star}$, add the edge $v_p^{i^\star} \to v_{(p+1) ~\mathrm{mod}~ r_{i^\star}}^{i^\star}$ to $\mathbb{E}$.
\State {\bf Return} $\mathcal{G} = (\mathbb{V},\mathbb{E})$.
\end{algorithmic}
\end{algorithm}

\textcolor{black}{Algorithm \ref{alg.BuildingGraphs} essentially tries to reverse the quotient process described in the proof of Theorem \ref{thm.LowerBound}. It starts with the quotient graph, given by the tree $\mathcal{T}$, and it constructs the original graph $\mathcal{G}$ be assigning nodes to each element of the partition. More precisely, the algorithm assigns to each node $j$ in the tree a set of $r_j$ nodes in $\mathcal{G}$, denoted by $V_j = \{v_p^j\}_{p=1}^{r_j}$. In step 3, the algorithm populates the graph $\mathcal{G}$ with edges corresponding to those found in the quotient $\mathcal{T}$, and it does so with the minimal possible amount of edges guaranteeing symmetry (as seen in the proof of Theorem \ref{thm.LowerBound}). An illustration of this step can be seen in Fig. \ref{fig.Step3}. In step 5, the algorithm adds a few more edges to guarantee the constructed graph $\mathcal{G}$ is connected, but in a way that does not effect the quotient process, as all new edges are between nodes in the same set $V_{i^\star}$ of the partition.} As the following theorem shows, choosing any path $\mathcal{T}$ will result in a graph of type OS$(r_1,\ldots,r_k)$. However, the number of edges in the graph depends on the path $\mathcal{T}$. We note that $\mathcal{T}$ must be a path rather than a general tree, see Remark \ref{rem.TreeOrPath} for further discussion.

\begin{figure}[b]
    \centering
    \vspace{-15pt}
    \includegraphics[width = 0.7\textwidth]{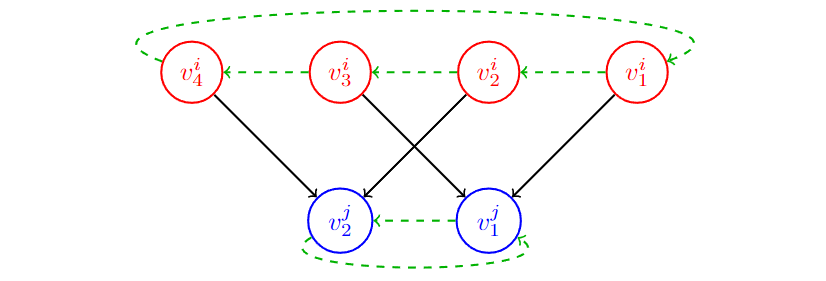}
    \vspace{-10pt}
    \caption{An illustration of step 3 in Algorithm \ref{alg.BuildingGraphs} with $r_i = 4$ (in red) and $r_j = 2$ (in blue). The algorithm starts with the edge between $v_1^i$ and $v_1^j$. It then moves along the green dashed lines, adding an edge after each one step. This step of the algorithm terminates when the algorithm tries to add an already existing edge, resulting in the black edges depicted in the figure.}
    \label{fig.Step3}
\end{figure}

\begin{thm}\label{thm.UpperBound}
Let $r_1,\ldots r_k$ be positive integers summing to $n$. For any path $\mathcal{T}$, Algorithm \ref{alg.BuildingGraphs} outputs a graph $\mathcal{G}$ of type OS$(r_1,\ldots,r_k)$ having $\sum_{\{i,j\}\in\mathcal{T}} \mathrm{lcm}(r_i,r_j) + \min_i r_i$ edges. Thus, there is a graph of type OS$(r_1,\ldots,r_k)$ having $M$ edges, where
\begin{align} \label{eq.UpperBound}
M = \min_{\mathcal{T}\text{\emph{ path on $k$ vertices}}} \sum_{\{i,j\}\in\mathcal{T}} \mathrm{lcm}(r_i,r_j) + \min_{i} r_i.
\end{align}
\end{thm}

\begin{proof}
We assume, without loss of generality and for the benefit of neater notation, that the path $\mathcal{T}$ is of the form $1 \to 2 \to \ldots \to k$, and let $i^\star$ be the vertex chosen in step 5, i.e. the node at which $r_i$ is minimized. By construction, the number of edges in $\mathcal{G}$ is equal to $\sum_{\{i,j\}\in\mathcal{T}} \mathrm{lcm}(r_i,r_j)$, plus $r_{i^\star}$ if $r_{i^\star} \ge 2$. Thus, it suffices to show the orbits of the action of $\mathrm{Aut}(\mathcal{G})$ on $\mathcal{G}$ are given by $V_1,\ldots, V_k$, where $V_j = \{v_j^p\}_{p=1}^{r_j}$, and that $\mathcal{G}$ is weakly connected. We start with the former.

Regarding the orbits, we must show that all nodes in $V_j$ are exchangeable, and that $V_j$ are invariant under $\mathrm{Aut}(\mathcal{G})$, for any $j=1,\ldots,k$. For the first claim, we consider the map $\psi:\mathbb{V}\to\mathbb{V}$ defined by sending each vertex $v^j_p$ to $v^j_{(p+1)~\mathrm{mod}~ r_j}$, illustrated by the dashed green edges in Fig. \ref{fig.Step3}. By construction, this map is an automorphism of $\mathcal{G}$. Moreover, iterating $\psi$ enough times would move $v^j_p$ to any vertex in $V_j$, hence any two nodes in $V_j$ are exchangable.

Second, we show that each $V_j$ is invariant under $\mathrm{Aut}(\mathcal{G})$, implying that  the orbits of $\mathrm{Aut}(\mathcal{G})$ in $\mathcal{G}$ are exactly $V_1,\ldots V_k$. This is obvious if $k=1$, as $V_1 = \mathbb{V}$. If $k\ge 2$, then either $i^\star\neq 1$ or $i^\star\neq k$ (or both). We assume $i^\star\neq k$ without loss of generality. Graph automorphisms preserve all graph properties, and in particular, they preserve the out-degree of vertices. As all edges are oriented from $V_j$ to $V_{j+1}$ or from $V_i^\star$ to itself. Therefore, the vertices in $V_k$ have an out-degree of $0$, and they are the only ones with this property. Thus $V_k$ must be invariant under $\mathrm{Aut}(\mathcal{G})$. The only vertices with edges to $V_k$ are in $V_{k-1}$, so $V_{k-1}$ is also invariant under $\mathrm{Aut}(\mathcal{G})$. Iterating this argument, we conclude that $V_1,\cdots V_k$ must all be invariant under the action of $\mathrm{Aut}(\mathcal{G})$, as claimed.

Now, we show that $\mathcal{G}$ is weakly connected. First, note the induced subgraph on $V_{i^\star}$ is weakly connected, due to the construction made in step 5. Indeed, this is clear if $r_{i^\star}=1$, and in the case $r_{i^\star} \ge 2$, the cycle $v^{i^\star}_1 \to v^{i^\star}_2 \to v^{i^\star}_3 \to \cdots$ eventually passes through all the nodes in $V_{i^\star}$. However, by the construction in step 3 of the algorithm, any two nodes $v^i_p$ and $v^j_p$ can be connected by a path. Thus any two arbitrary vertices $v^{j_1}_{p_1}$ and $v^{j_2}_{p_2}$ can be linked as follows - first, go from $v^{j_1}_{p_1}$ to $v^{i}_{p_1~\mathrm{mod}~ r_i}$ using the edges added in step 3; then, move to $v^{i}_{p_2~\mathrm{mod}~ r_{i}}$ using the edges added in step 5; lastly, continue from $v^{i}_{p_2~\mathrm{mod}~ r_{i}}$ to $v^{j_2}_{p_2}$ using the edges added to the graph in step 3. Hence, $\mathcal{G}$ is weakly connected, and the proof is complete.
\end{proof}

\begin{rem} \label{rem.TreeOrPath}
The lower bound considers all possible trees, but the upper bounds only considers path graphs. It can be seen in the proof of Theorem \ref{thm.UpperBound} that the fact that $\mathcal{T}$ is a path is only used to prove that each $V_j$ is invariant under the action of $\mathrm{Aut}(\mathcal{G})$. This might be false if $\mathcal{T}$ is any tree, as the following example shows.  Consider Algorithm \ref{alg.BuildingGraphs} with 4 clusters of size $1$ and a tree $\mathcal{T}$ as depicted in Fig. \ref{fig,Tree}. In this case, the graph $\mathcal{G}$ is equal to the tree $\mathcal{T}$. However, the permutation switching the nodes $2$ and $3$ is a graph automorphism, so there is a cluster of size at least $2$, hence $\mathcal{G}$ is not OS$(1,1,1,1)$. Nevertheless, one should notice that the upper and lower bounds coincide when the number of clusters $k$ is at most $3$, as any tree on at most $3$ nodes must be a path.
\end{rem}

\begin{figure}[b]
        \centering
        \vspace{-10pt}
        \includegraphics[width=0.4\textwidth]{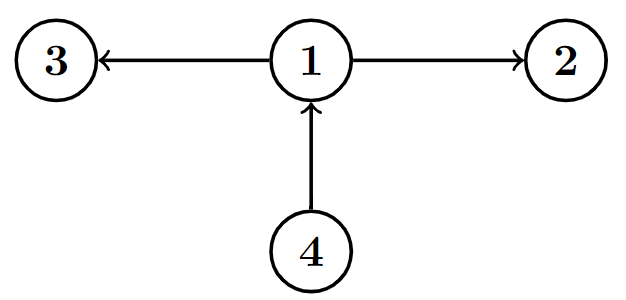} 
        \vspace{-10pt}
        \caption{The tree graph discussed in Remark \ref{rem.TreeOrPath}.}
        \label{fig,Tree}
\end{figure}

\begin{rem}
The lower bound \eqref{eq.LowerBound} can be found using Kruskal's algorithm, which finds a minimal spanning tree in polynomial time \cite{Cormen2009}. However, the upper bound \eqref{eq.UpperBound} requires one to solve a variant of the traveling salesman problem, which is {\bf NP}-hard \cite{Cormen2009}. 
\end{rem}


Algorithm \ref{alg.BuildingGraphs} solves the general cluster assignment problem, as it constructs graphs of type OS$(r_1,\ldots,r_k)$ for any cluster sizes $r_1,\ldots,r_k$. Unfortunately, the bound \eqref{eq.UpperBound} is implicit in terms of the number of nodes and clusters. We elucidate it by applying it to more specific cases, resulting in concrete bounds on the number of edges needed for clustering in these cases.
\begin{cor}
Suppose all cluster sizes $r_1,\ldots,r_k > 1$ are equal. Then there exists a graph of type OS($r_1,\ldots,r_k$) with at most $n = r_1+\cdots+r_k$ edges, and this number of edges is minimal.
\end{cor}

\begin{proof}
Let $r$ be the size of all clusters, so that $\mathrm{lcm}(r,r) = r$ and the number of clusters is $k=n/r$. Thus, the graph $\mathcal{G}$  outputted by Algorithm \ref{alg.BuildingGraphs} (for an arbitrary $\mathcal{T}$) has exactly $n$ edges, as the summation over the edges has $k-1$ elements. It remains to show that no graph of type OS($r_1,\ldots,r_k$) can have fewer than $n$ edges. As any graph with fewer than $n-1$ edges is not weakly connected \cite{Godsil2001}, it suffices to prove that such a graph cannot have exactly $n-1$ edges. 

First, note that the out-degree is preserved by the action of $\mathrm{Aut}(\mathcal{G})$, meaning that vertices in the same cluster have the same out-degree. Denoting the out-degree of nodes in the $i$-th cluster by $d_i$, the total number of edges is equal to the sum of the out-degree over all nodes, i.e. to $r(d_1+\cdots+d_k)$. In particular, the number of edges, $n-1$ is divisible by $r$. As $n=kr$, $n$ is also divisible by $r$, which together implies that $r$ divides $1$, which is absurd. Thus, no such graph on $n-1$ edges can exist.
\end{proof}

\begin{cor}
Let $r_1,\ldots,r_k$ be positive integers such that $k\ge 2$ and that for every $j,l$, either $r_j$ divides $r_l$ or vice versa. Then there exists a graph of type OS($r_1,\ldots,r_k$) with $n = r_1+\cdots+r_k$ edges.
\end{cor}

\begin{proof}
We reorder the numbers $r_1,\cdots, r_k$ so that $r_l$ divides $r_j$ for $l\le j$. We note that if $r_l$ divides $r_j$, then $\mathrm{lcm}(r_l,r_j) = r_l$. Thus, running Algorithm \ref{alg.BuildingGraphs} with $\mathcal{T} = 1\to 2\to\ldots\to k$ gives a graph type OS($r_1,\ldots,r_k$) with the following number of edges:
$$
\sum_{j = 1}^{k-1} \mathrm{lcm}(r_j,r_{j+1}) + r_1 = \sum_{j=1}^{k-1} r_{j+1} + r_1=\sum_{j=1}^k r_j = n.
$$
\end{proof}

\begin{cor}
Let $r_1,\ldots,r_k$ be positive integers such that $r_j \le q$ for all $j$, and let $n = r_1+\cdots+r_k$. Then there exists a graph of type OS($r_1,\ldots,r_k$) with at most $n+O(q^3)$ edges.
\end{cor}

\begin{proof}
We assume without loss of generality that $r_1\le r_2\le \cdots \le r_k$. Let $m_l$ be the number of clusters of size $l$ for $l=1,2,\ldots,q$, and let $\mathcal{G}$ be the graph constructed by Algorithm \ref{alg.BuildingGraphs} for the path $\mathcal{T} = 1\to 2\to\cdots\to k$. If $r_j=r_{j+1}$ then $\mathrm{lcm}(r_j,r_{j+1}) = r_j$, and $\mathrm{lcm}(r_j,r_{j+1}) \le r_j r_{j+1}$ otherwise. Thus, the number of edges in $\mathcal{G}$ is given by:
\begin{align*}
\sum_{j =1}^{k-1} \mathrm{lcm}(r_j,r_{j+1}) +r_1 \le  \hspace{-7pt}\sum_{\substack{l\in\{1,\ldots,q\},\\ m_l \neq 0}} \hspace{-7pt}(m_l-1)l + \sum_{l=1}^{q-1} l(l-1) + r_1.
\end{align*}
Indeed, for each $l\in\{1,\ldots,q\}$, if there's at least one cluster of size $l$, then there are $m_l -1$ edges in the path $\mathcal{T}$ that touch two clusters of size $l$. The second term bounds the number of edges between clusters of different sizes. We note that $n = \sum_{l=1}^q l m_l$, so the first term is bounded by $n$. As for the second term, we write $l(l-1) \le l^2$ and use the formula $\sum_{l=1}^{q-1} l^2 = \frac{(q-1)q(2q-1)}{6}$. 
Lastly, the last term $r_1$ is bounded by $q$. This completes the proof.
\end{proof}

Theorem \ref{thm.GlobalUpperBound} will show that the upper bound \eqref{eq.UpperBound} is bounded by $\frac{(k-1)n^2}{k^2}$ for any cluster sizes, and it will also give a heuristic for choosing the path $\mathcal{T}$ for Algorithm \ref{alg.BuildingGraphs} guaranteeing the number of edges does not exceed this upper bound. 

\vspace{-15pt}
\subsection{Robust OS-type graphs} \label{subsec.Robust}
The previous section presented a solution to the problem of cluster assignment. Namely, given a collection of homogeneous agents and the desired cluster sizes, the designer constructs the interconnection graph using Algorithm \ref{alg.BuildingGraphs}, and then chooses a controller following Assumption \ref{assump.3}. The analysis depicted in \cite{Sharf2019b} shows that the closed-loop network would then converge to the desired clustered steady-state. However, there are no guarantees on what happens if the network changes abruptly, either due to hardware or software malfunction, a cyber attack, or both. In these cases, one (or more) of the agents effectively become disconnected from the rest of the network, and are effectively removed from the dynamical system and the interaction graph. For this reason, we wish to explore the robustness of OS-type graphs to changes. Ideally, when one (or more) agents are removed from the system due to a malfunction, all other agents should still cluster as before. More specifically, two non-compromised agents that were previously in the same cluster, should still belong to the same cluster. We thus make the following definition:

\begin{defn}
The oriented graph $\mathcal{G} =(\mathbb{V},\mathbb{E})$ with $n$ nodes is said to be \emph{$s$-robustly OS$(r_1,\ldots,r_k)$} for a positive integer $s$ (called the clustering robustness parameter) if the following conditions hold: i) $\mathcal{G}$ is weakly connected; ii) The orbits $V_1,\ldots,V_k$ of the action of $\mathrm{Aut}(\mathcal{G})$ on $\mathcal{G}$ have sizes $r_1,\ldots,r_k$ respectively; and iii) for any set $A\subseteq \mathcal{V}$ such that $|A| \ge n - s$ (comprised of the non-compromised agents), denoting the induced subgraph with node set $A$ as $\tilde{\mathcal{G}}$, the action of $\mathrm{Aut}(\mathcal{\tilde G})$ on $\mathcal{\tilde G}$ has orbits $V_1\cap A, \ldots, V_k \cap A$. Moreover, we say that the oriented graph $\mathcal{G}$ is \emph{totally robustly OS$(r_1,\ldots,r_k)$} if it is $s$-robustly OS$(r_1,\ldots,r_k)$ for $s = n$. 
\end{defn}

\begin{prop} \label{prop.robust}
Let $r_1,\ldots,r_k$ be positive integers, and let $\mathcal{G}= (\mathbb{V},\mathbb{E})$ be a $1$-robustly OS$(r_1,\ldots,r_k)$ with clusters $V_1,\ldots,V_k$. Then for any $i\neq j \in \mathcal{V}$, the induced bi-partite subgraph $\mathcal{G}_{ij}$ is either empty or complete.
\end{prop}
\begin{proof}
Suppose that $\mathcal{G}_{ij}$ is non-empty, and let $v_i \in V_i$, $v_j\in V_j$ be two connected nodes in it. We claim all vertices in $V_j$ are connected to $v_i$. Indeed, consider the graph $\tilde{\mathcal{G}} \setminus \{v_i\}$ corresponding to the case in which $v_i$ malfunctions. By assumption, both the actions of $\mathrm{Aut}(\mathcal{G})$ and $\mathrm{Aut}(\tilde{\mathcal{G}})$ can send any node in $V_j$ to any other node in $V_j$. In particular, all nodes in $V_j$ share both their $\mathcal{G}$-degree and their $\tilde{\mathcal{G}}$-degree. However, the $\tilde{\mathcal{G}}$-degree of $v_j$ is smaller than its $\mathcal{G}$-degree by one, so this must also be the case for any other node in $V_j$, meaning that node must be connected to $v_i$. Reiterating this argument (while changing $i$ with $j$ and the edge $\{v_i,v_j\}$ with the edge $\{v_i,v_j^\prime\}$ for $v_j^\prime \in V_j$) shows that any node in $V_i$ is connected to any node in $V_j$. 
\end{proof}

Proposition \ref{prop.robust} suggests a necessary update to Algorithm \ref{alg.BuildingGraphs} to get robust OS-type graphs. Indeed, the modulo-based construction of edges between different clusters is replaced by taking all possible edges. The adapted Algorithm \ref{alg.BuildingRobustGraphs} is given below.

\begin{algorithm} [t!]
\caption{Building totally robust OS-type graphs}
\label{alg.BuildingRobustGraphs}
{\bf Input:} A collection $r_1,\ldots,r_k$ of positive integers summing to $n$, and a path $\mathcal{T}$ on $k$ vertices.\\
{\bf Output:} A graph $\mathcal{G}$ of type OS$(r_1,\ldots,r_k)$.
\begin{algorithmic}[1]
\State If $k=1$, return the complete graph on $n$ nodes. Otherwise, continue.
\State Let $\mathcal{G} = (\mathbb{V},\mathbb{E})$ be an empty graph.
\State For any $j=1,\ldots,k$ and $p=1,\ldots,r_j$, add a node with label $v^j_p$ to $\mathbb{V}$.
\State For any edge $\{i,j\}$ in $\mathcal{T}$, $p_i = 1,\ldots,r_i$ and $p_j = 1\ldots,r_j$, add the edge $v^i_p \to v^j_q$  to $\mathbb{E}$.
\State {\bf Return} $\mathcal{G} = (\mathbb{V},\mathbb{E})$.
\end{algorithmic}
\end{algorithm}

\begin{thm}\label{thm.robust}
Let $r_1,\ldots r_k$ be positive integers, and let $n = r_1+\cdots+r_k$. For any path $\mathcal{T}$, Algorithm \ref{alg.BuildingRobustGraphs} outputs a graph $\mathcal{G}$ which is totally robustly OS$(r_1,\ldots,r_k)$, having $\sum_{\{i,j\}\in\mathcal{T}} r_ir_j$ edges.
\end{thm}

\begin{proof}
We denote $V_i = \{v_p^i\}_{p=1}^{r_i}$. We first prove $\mathcal{G}$ is an OS-type graph. It is obviously weakly connected as $r_i \ge 1$ for all $i$, and one can prove that the sets $V_i$ are $\mathrm{Aut}(\mathcal{G})$-invariant similarly to the proof of Theorem \ref{thm.UpperBound}. Moreover, one can see that the combination of any permutations on $V_i$ for $i=1,\ldots,k$ gives an automorphism of $\mathcal{G}$, meaning that all nodes inside $V_i$ are equivalent under $\mathrm{Aut}(\mathcal{G})$. In particular, we conclude that $\mathcal{G}$ is an OS$(r_1,\ldots,r_k)$ graph.
Now, let $A$ be a set of non-compromised agents, and let $c_i = |V_i \setminus A|$ be the number of compromised agents in cluster $i$. We observe that the induced subgraph $\tilde{\mathcal{G}}$ on $A$ is identical to the output of the algorithm for the path $\mathcal{T}$ and cluster sizes $r_1-c_1,\ldots,r_k-c_k$, and in particular, the first part of this proof shows that the clusters in $\tilde{\mathcal{G}}$ are exactly $V_1\cap A,\ldots, V_k\cap A$. As the number of edges in $\mathcal{G}$ is obviously $\sum_{\{i,j\}\in\mathcal{T}} r_ir_j$, the proof is complete.
\end{proof}

Regrading sparsity, Proposition \ref{prop.robust} and Algorithm \ref{alg.BuildingRobustGraphs} show that the number of edges in $s$-robust (or totally robust) OS-type graphs can be bounded from above by $M^\prime$ and from below by $m^\prime$, where:
\begin{align} \label{eq.RobustBounds}
    m^\prime = \min_{\mathcal{T}\text{\emph{ tree on $k$ vertices}}} \sum_{\{i,j\}\in\mathcal{T}} r_ir_j, \quad M^\prime = \min_{\mathcal{T}\text{\emph{ path on $k$ vertices}}} \sum_{\{i,j\}\in\mathcal{T}} r_ir_j.
\end{align}

Furthermore, the relation $ab = \gcd(a,b)\mathrm{lcm}(a,b)$ gives a connection between \eqref{eq.RobustBounds}, \eqref{eq.LowerBound} and \eqref{eq.UpperBound}. Namely, if $\rho = \max_{i,j} \gcd(r_i,r_j)$, then $m\le m^\prime \le \rho m$ and $M \le M^\prime \le \rho M$. We note that $\rho \le \max_i r_i$, meaning the number of edges required to get totally robust OS-type graphs isn't much larger than the number of edges required to get OS-type graph, at least when there are no large clusters.

Before moving on to a numerical example, we wish to provide one more result relating the upper bound $M^\prime$ (and hence $M$) to the number of nodes and the number of clusters.
\begin{thm} \label{thm.GlobalUpperBound}
Let $r_1,\ldots,r_k$ be positive integers summing to $n$. Suppose without loss of generality that $r_1 \ge r_2 \ge \cdots \ge r_k$, and let $\mathcal{T}$ be the path $1\to k\to 2\to (k-1)\to\cdots$ on $k$ nodes. The graph $\mathcal{G}$ outputted by Algorithm \ref{alg.BuildingRobustGraphs} has no more than $\frac{(k-1)n^2}{k^2}$ edges.
\end{thm}

\begin{proof}
We first consider the case in which $k$ is odd, i.e., $k = 2\ell+1$ for some integer $\ell$. There are two types of edges in the path $\mathcal{T}$ - edges of the form $i\to (k+1)-i$ for $i=1,2,\ldots, \ell$, and edges of the form $(k+1)-i \to i+1$ for $i=1,2,\ldots,\ell$. Thus, by Theorem \ref{thm.robust}, the number of edges in $\mathcal{G}$ is no larger than the value of the following (continuous) optimization problem:
\begin{align*}
    \nu = \max \left\{\sum_{\{i,j\}\in\mathcal{T}} x_ix_j = \sum_{i=1}^\ell x_i x_{(k+1-i)} + \sum_{i=1}^{\ell-1} x_{i+1} x_{(k+1-i)} : \sum_{i=1}^k x_i = n,~x_1 \ge x_2 \ge \ldots \ge x_k \right\}.
\end{align*}

We let $x_1^\star,\ldots,x_k^\star$ be the optimal solution to the problem above. If we show that all $x_i^\star$ are equal to each other (and hence to $n/k$), this would imply that the number of edges in $\mathcal{G}$ is bounded by the value of the cost function at $x_1^\star,\ldots,x_k^\star$, which is equal to $\nu = \frac{(k-1)n^2}{k^2}$, as claimed. Suppose, for example and heading toward contradiction, that $x_2^\star < x_3^\star$ and that $\ell \ge 3$. The derivative of the cost function $F$ satisfies the following inequality:
\begin{align*}
    \frac{\partial F}{\partial x_2}|_{x^\star} = x_k^\star + x_{k-1}^\star \le x_{k-1}^\star + x_{k-2}^\star = \frac{\partial F}{\partial x_3}|_{x^\star},
\end{align*}
where the inequality follows from the constraints of the optimization problem. Thus, slightly reducing $x_2^\star$ and increasing $x_3^\star$ by the same amount results in a feasible solution with at least the same value of the cost function. This is contradictory to the manner in which $x^\star$ was chosen, meaning that $x_2^\star = x_3^\star$. A similar argument shows that in fact, $x_1^\star = x_2^\star = \cdots = x_\ell^\star$, and that $x_{\ell+1}^\star = \cdots = x_k^\star$. Finally, if the value of the former is different from the value of the latter, one similarly shows that reducing all $x_1^\star,\ldots,x_\ell^\star$ by some small $\epsilon > 0$ and simultaneously increasing all $x_{\ell+1}^\star,\ldots,x_k^\star$ by the same $\epsilon$ gives a feasible solution with at least the same value of the cost function. Therefore, all the $x_i^\star$-s must be equal, and the bound is achieved. The proof for an even number of clusters $k$ is analogous, and is omitted for the sake of brevity and due to lack of space.
\end{proof}

\begin{rem} \label{rem.WorstCaseSparse}
The proof of Theorem \ref{thm.GlobalUpperBound} also gives the worst case scenario for Algorithm \ref{alg.BuildingRobustGraphs}, namely, the maximal number of edges is achieved when all clusters are equally-sized. 
The upper bound in Theorem \ref{thm.GlobalUpperBound} also applies to the non-robust graphs outputted by Algorithm \ref{alg.BuildingGraphs}. In fact, this upper bound is tight, at least in order of magnitude. Indeed, a graph with $O\left(\frac{(k-1)n^2}{k^2}\right)$ edges can be found by taking $r_1,\ldots,r_k$ as the $(L+1),\ldots,(L+k)$-th prime numbers, where the $L$ is chosen as an integer satisfying $L\log L \approx \frac{n}{k}$.
\end{rem}

\section{Numerical Example}
We consider a collection of identical agents, all of the form $\dot{x} = -x+u+\alpha ,y=x$ where $\alpha$ is a log-uniform random variable between $0.1$ and $1$, identical for all agents. In all experiments described below, we considered identical controllers, equal to the static nonlinearity of the form
$$
\mu = a_1 + a_2 (\zeta + \cos(\zeta)),
$$
where $a_1$ was chosen as a Gaussian random variable with mean $0$ and standard deviation $10$, and $a_2$ was chosen as a log-uniform random variable between $0.1$ and $10$. We note that the agents are indeed output-strictly MEIP and the controllers are MEIP. Moreover, the network is homogeneous, so $\mathrm{Aut}(\mathcal{G},\Sigma,\Pi) = \mathrm{Aut}(\mathcal{G})$. Thus, we can use the graphs constructed Algorithm \ref{alg.BuildingGraphs}, to force a clustering behavior.

We first consider a network of $n=15$ agents and tackle the cluster synthesis problem with five equally-sized clusters, i.e., $r_1=r_2=r_3=r_4=r_5=3$. One possible graph forcing these clusters, as constructed by Algorithm  \ref{alg.BuildingGraphs} for a path $\mathcal{T}$ of the form $1\to 2\to 3\to 4\to 5$, can be seen in Fig. \ref{fig.ClusterSynthesisTheoremGraph1}, along with the agents' trajectories for the closed-loop system. 

Second, we consider a network of $n=10$ agents with desired cluster sizes $r_1=1,r_2=2,r_3=3,r_4=4$. We build a graph forcing these cluster sizes by using Algorithm  \ref{alg.BuildingGraphs} for a path $\mathcal{T}$ of the form $4\to 2\to 1\to 3$, which is the minimizer in \eqref{eq.UpperBound}. The graph can be seen in Fig. \ref{fig.ClusterSynthesisTheoremGraph1}, along with the agents' trajectories for the closed-loop system. 
 \begin{figure}[t!]
 \begin{center}
 	\hspace{.5cm}\subfigure[Graph forcing cluster sizes $r_1=r_2=r_3=r_4=r_5=3$. Nodes with the same color will be in the same cluster.] {\scalebox{.80}{\includegraphics{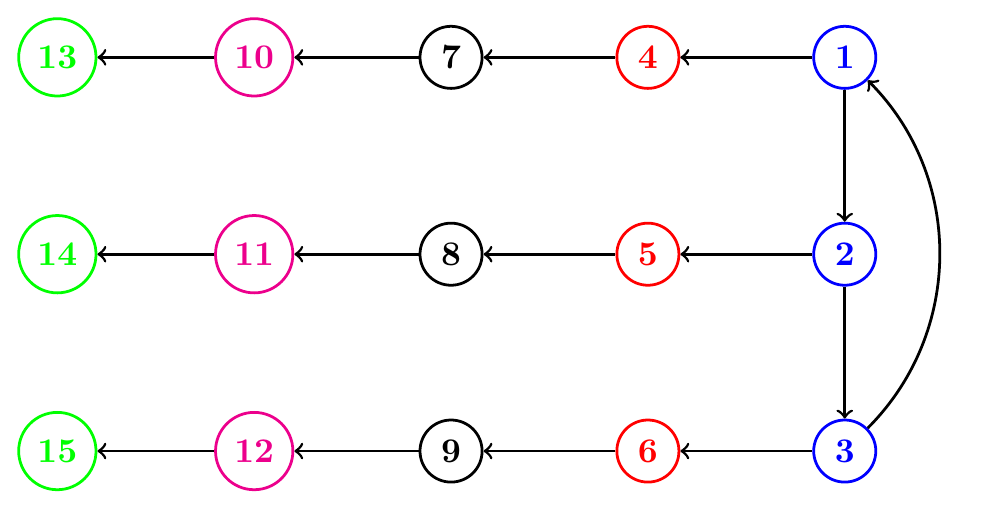}}} \hspace{0.5cm}
	\subfigure[Agent's trajectories for the closed-loop system. Colors correspond to node colors in the graph.] {\scalebox{.5}{\includegraphics{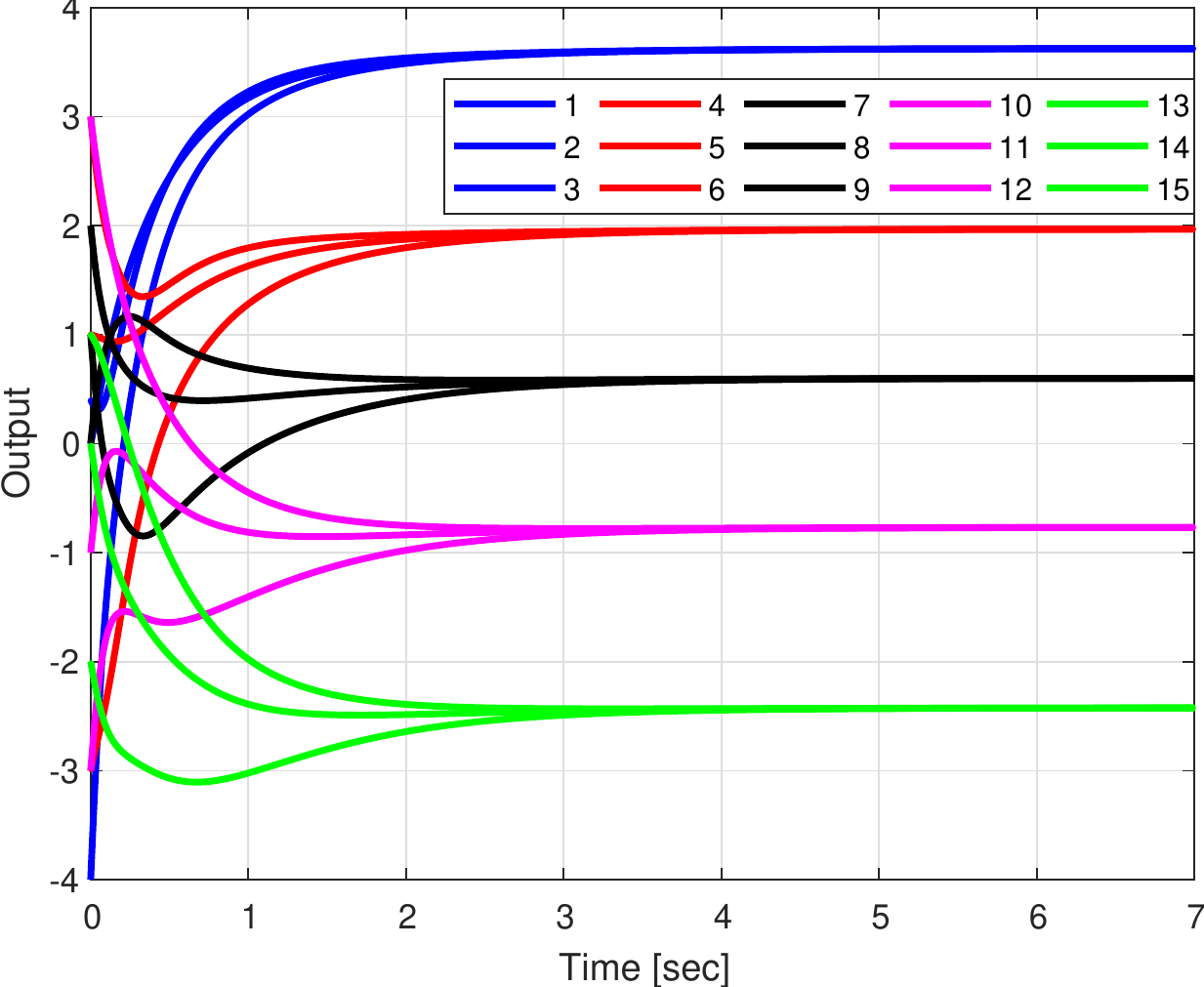}}} 
   \caption{First example of graphs solving the cluster synthesis problem, achieved by running Algorithm \ref{alg.BuildingGraphs}.}\hspace{.5cm}
	\label{fig.ClusterSynthesisTheoremGraph1}
	\vspace{-25pt}
 \end{center}
 \end{figure}
 \begin{figure}[t!]
 \begin{center}
 	\hspace{0.5cm}\subfigure[Graph forcing cluster sizes $r_1=1,r_2=2,r_3=3,r_4=4$. Nodes with the same color will be in the same cluster.] {\scalebox{.80}{\includegraphics{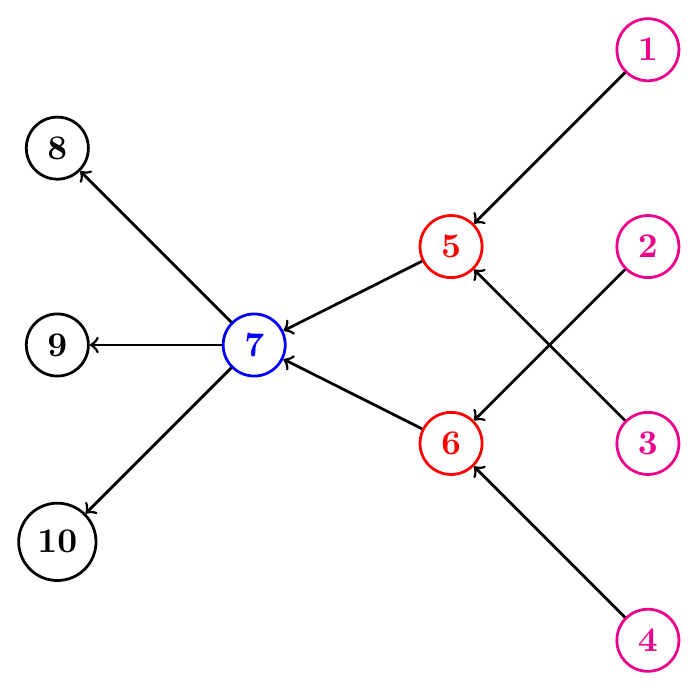}}}\hspace{0.5cm} \vspace{0.2cm}
	\subfigure[Agent's trajectories for the closed-loop system. Colors correspond to node colors in the graph.] {\scalebox{.5}{\includegraphics{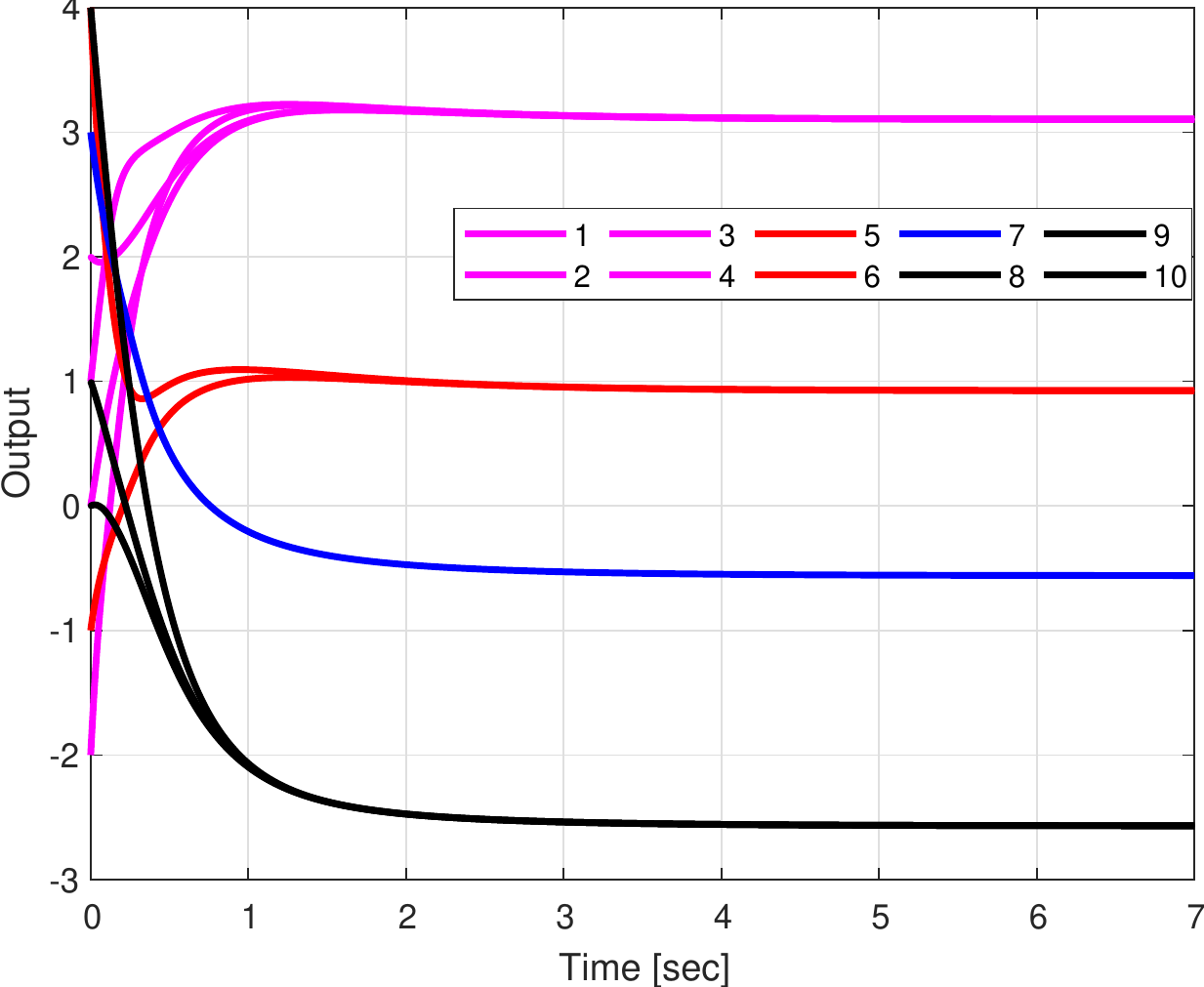}}} \hspace{0.5cm}
   \caption{First example of graphs solving the cluster synthesis problem, achieved by running Algorithm \ref{alg.BuildingGraphs}.}
	\label{fig.ClusterSynthesisTheoremGraph2}
	\vspace{-15pt}
 \end{center}
 \end{figure}

\section{Conclusions}\label{sec.conclusion}
This work explored the problem of cluster assignment for homogeneous diffusively-coupled multi-agent systems. We relied upon the clustering analysis results of \cite{Sharf2019b} to exhibit synthesis procedures that guarantee a clustering behaviour, no matter the desired number of clusters nor their size. This was done by prescribing graph synthesis algorithms which have certain symmetry properties that reflect the desired clustering assignment. When such graphs are used in a network comprised of weakly equivalent agent and controller dynamics, the network converges to a cluster configuration. We further studied the robustness of such MAS to node malfunctions, and presented a graph synthesis procedure which guarantees the clustering structure is robust to any number of agent malfunctions. The results were demonstrated in a numerical example.

\vspace{-15pt}
\bibliography{main}%



\end{document}